\documentclass[conference]{IEEEtran}
\IEEEoverridecommandlockouts
\usepackage{cite}
\usepackage{amsmath,amssymb,amsfonts}
\usepackage{amsthm}
\usepackage{algorithmic}
\usepackage{graphicx}
\usepackage{textcomp}
\usepackage{xcolor}

\usepackage{bm}
\usepackage{subfigure}
\usepackage{algorithmic}
\usepackage[lined,boxed,commentsnumbered, ruled, linesnumbered]{algorithm2e}

\newtheorem{lemma}{\bf Lemma}
\newtheorem{assumption}{\bf Assumption}
\newtheorem{theorem}{\bf Theorem}

\def\BibTeX{{\rm B\kern-.05em{\sc i\kern-.025em b}\kern-.08em
    T\kern-.1667em\lower.7ex\hbox{E}\kern-.125emX}}
\begin{document}

\title{Wireless Federated Learning over MIMO Networks: Joint Device Scheduling and Beamforming Design
}

\author{\IEEEauthorblockN{$\text{Shaoming Huang}^{*\dag\ddagger}$, $\text{Pengfei Zhang}^{*\dag\ddagger}$, $\text{Yijie Mao}^{*}$, $\text{Lixiang Lian}^{*}$, and $\text{Yuanming Shi}^{*}$}
\IEEEauthorblockA{${}^{*}$\text{School of Information Science and Technology, ShanghaiTech University, Shanghai $201210$, China} \\
${}^{\dag}$\text{Shanghai Institute of Microsystem and Information Technology, Chinese Academy of Sciences, China}\\
${}^{\ddagger}$\text{University of Chinese Academy of Sciences, Beijing $100049$, China}\\
E-mail:  huangshm@shanghaitech.edu.cn, pfzhang1003@163.com, \{maoyj, lianlx, shiym\}@shanghaitech.edu.cn}
}

\maketitle

\begin{abstract}
Federated learning (FL) is recognized as a key enabling technology to support distributed artificial intelligence (AI) services in future 6G. By supporting decentralized data training and collaborative model training among devices, FL inherently tames privacy leakage and reduces transmission costs. Whereas, the performance of the wireless FL is typically restricted by the communication latency. Multiple-input multiple-output (MIMO) technique is one promising solution to build up a communication-efficient edge FL system with limited radio resources. In this paper, we propose a novel joint device scheduling and receive beamforming design approach to reduce the FL convergence gap over shared wireless MIMO networks. Specifically, we theoretically establish the convergence analysis of the FL process, and then apply the proposed device scheduling policy to maximize the number of weighted devices under the FL system latency and sum power constraints. Numerical results verify the theoretical analysis of the FL convergence and exhibit the appealing learning performance of the proposed approach.
\end{abstract}

\section{Introduction}
With the large-scale deployment of 5G around the world, 6G has attracted increasing attention form both industry and academia, which is envisioned as an unprecedented evolution from ``connected things'' to ``connected intelligence'', thereby forming the backbone of a hyper-connected cyber-physical world with the integration of humans, things and intelligence \cite{letaief2019roadmap, shi2020communication, saad2020vision, tataria20216g}. Recently, various artificial intelligence (AI) applications for 6G have emerged and penetrated in almost all verticals such as sustainable cities \cite{jararweh2020trustworthy}, industrial Internet of Things (IoT) \cite{lopez2021massive}, e-health services \cite{mucchi20206g}, etc. Traditional centralized machine learning (ML) frameworks, which require a cloud center to store and process the raw data collected from devices, are becoming impractical for enormous privacy-sensitive data and low-latency communication requirements \cite{zhu2020toward}. All the aforementioned reasons stimulate the development of federated learning (FL), where devices with limited hardware resources respectively train their local models with their raw datasets, and only the local models are transmitted from devices to the edge server for global model aggregation \cite{mcmahan2017communication}.

The deployment of FL over real wireless networks still faces significant challenges, among which the communication latency is becoming the major bottleneck with the rapid advances in computational capability. Considerable researches have been devoted to address this issue for both analog and digital FL systems. In analog FL systems, over-the-air computation (AirComp) technology is broadly leveraged to implement the efficient concurrent transmission of locally computed updates by exploiting the superposition property of wireless multiple access channels (MAC) \cite{yang2020federated, zhu2020broadband, amiri2020federated}. For multiple-input single-output (MISO) AirComp, to better trade off the learning performance and the communication efficiency, \cite{yang2020federated} proposed a joint device scheduling and receive beamforming design approach to maximize scheduled devices, and \cite{zhu2020broadband} proposed a broadband analog aggregation scheme which enables linear growth of the latency-reduction ratio with the device population. Furthermore, a distributed stochastic gradient descent algorithm was implemented for a bandwidth-limited fading MAC \cite{amiri2020federated}. However, the model aggregation performance of MISO AirComp is severely limited by the unfavorable wireless propagation channel. To build up a communication-efficient edge FL system, multiple-input multiple-output (MIMO) technique has been widely recognized as a promising way to support high-reliability for massive device connectivity as well as high-accuracy and low-latency for model aggregation via exploiting spatial degree of freedom \cite{wang2020federated, yang2020federatedieeenet, liu2021reconfigurable}.

Another line of works concentrates on the digital FL to circumvent the strict synchronization requirement at the symbol level and inherent corruption from channel noise during model aggregation stage in the analog FL system \cite{tran2019federated, chen2021convergence, shi2020joint, li2021delay}. \cite{tran2019federated} investigated the trade-off between the convergence time and device energy consumption via Pareto efficiency model. \cite{chen2021convergence} developed a probabilistic device scheduling scheme and first used artificial neural networks for the prediction of model parameters to minimize the convergence time. \cite{shi2020joint} optimized the convergence rate with a given time budget via a joint device scheduling and resource allocation. In \cite{li2021delay}, a theoretical analysis of the distributions of the per-round delay and overall delay were characterized, respectively. However, prior works on digital FL only are limited to the situation where each device transmits its local updates to the single-antenna edge server with orthogonal multiplexing approaches, e.g., time-division multiple access (TDMA) \cite{tran2019federated} and frequency-division multiple access (FDMA) \cite{chen2021convergence, shi2020joint, li2021delay}. In contract, multiple-antenna technique has shown its apparent performance gain in terms of convergence time \cite{vu2020cell}. To the best of our knowledge, there is still a lack of investigation on the digital FL systems over shared wireless MIMO networks.

In this paper, we consider a delay-aware digital FL system, where the edge server equipped with multiple antennas orchestrates the learning process via exchanging model parameters with single-antenna devices in shared wireless channels. At first, we theoretically establish the convergence analysis of the FL process after fixed rounds based on reasonable assumptions, which shows that as the weighted sum of scheduled devices of each round increases, the convergence optimality gap will be decreased. In order to minimize the convergence optimality gap, we design a novel device scheduling policy to allow as many weighted devices as possible to participate in the training process. Concretely speaking, we optimize receive beamforming design to obtain a device priority, and then iteratively add a device to the scheduling set based on the device priority when corresponding constraints of the wireless FL system satisfied. Numerical experiments are conducted to validate the theoretical analysis and demonstrate the superior performance of the proposed scheduling policy.

 \section{System Model}
 \subsection{Federated Learning Model} \label{FL_Model}
We consider the canonical FL system consisting of one edge server and $K$ devices indexed by $\mathcal{K} = \{1, \ldots, K\}$, which aims to collaboratively learn a common model \cite{mcmahan2017communication}. Each device $k \in \mathcal{K}$ has its own local dataset, of which the raw data are unavailable to other devices and the edge server. To facilitate the learning, the goal of the FL process is usually expressed as the following optimization problem:
\begin{equation}
 \underset{\bm{w} \in \mathbb{R}^d}{\text{minimize}} \  F(\bm{w}) := \sum_{k \in \mathcal{K}} \alpha_k \mathbb{E}_{\bm{\xi} \sim \mathcal{D}_k}[f(\bm{w};\bm{\xi})],
 \label{main_problem}
\end{equation}
where $\bm{w} \in \mathbb{R}^d $ represents the global model parameters, data sample $\bm{\xi}$ obeys a certain probability distribution $ \mathcal{D}_k $, and the global loss function $F(\cdot)$ is the weighted sum of the local loss functions $ f(\cdot;\bm{\xi}) $, where the weight factor $ \alpha_k > 0 $ satisfies $ \sum_{k \in \mathcal{K}} \alpha_k = 1 $. Basically,  $\alpha_k$ can be set as $ n_k / \sum_{k \in \mathcal{K}} n_k $, where $ n_k $ is the number of data samples at device $k$ \cite{li2020federated}.

The canonical FL runs in synchronized rounds of computation and communication process between the edge server and devices, consisting of three stages:
\begin{enumerate}
\item [1)] \emph{Global Model Dissemination: }The edge server first decides which devices to participate in the current round, and the set of scheduled devices at the round $t$ is denoted as $\mathcal{S}_{t}$. Then, the edge server broadcasts the global model parameters to all scheduled devices. In this paper, the scheduling policy is based on whether the devices satisfy the system latency constraints, which will be presented in Section \ref{System Optimization}. If there is no device scheduled at the current round, edge server will postpone this round until $\mathcal{S}_{t} \neq \varnothing$.

\begin{figure}[ht]
	\includegraphics[width=0.5\textwidth]{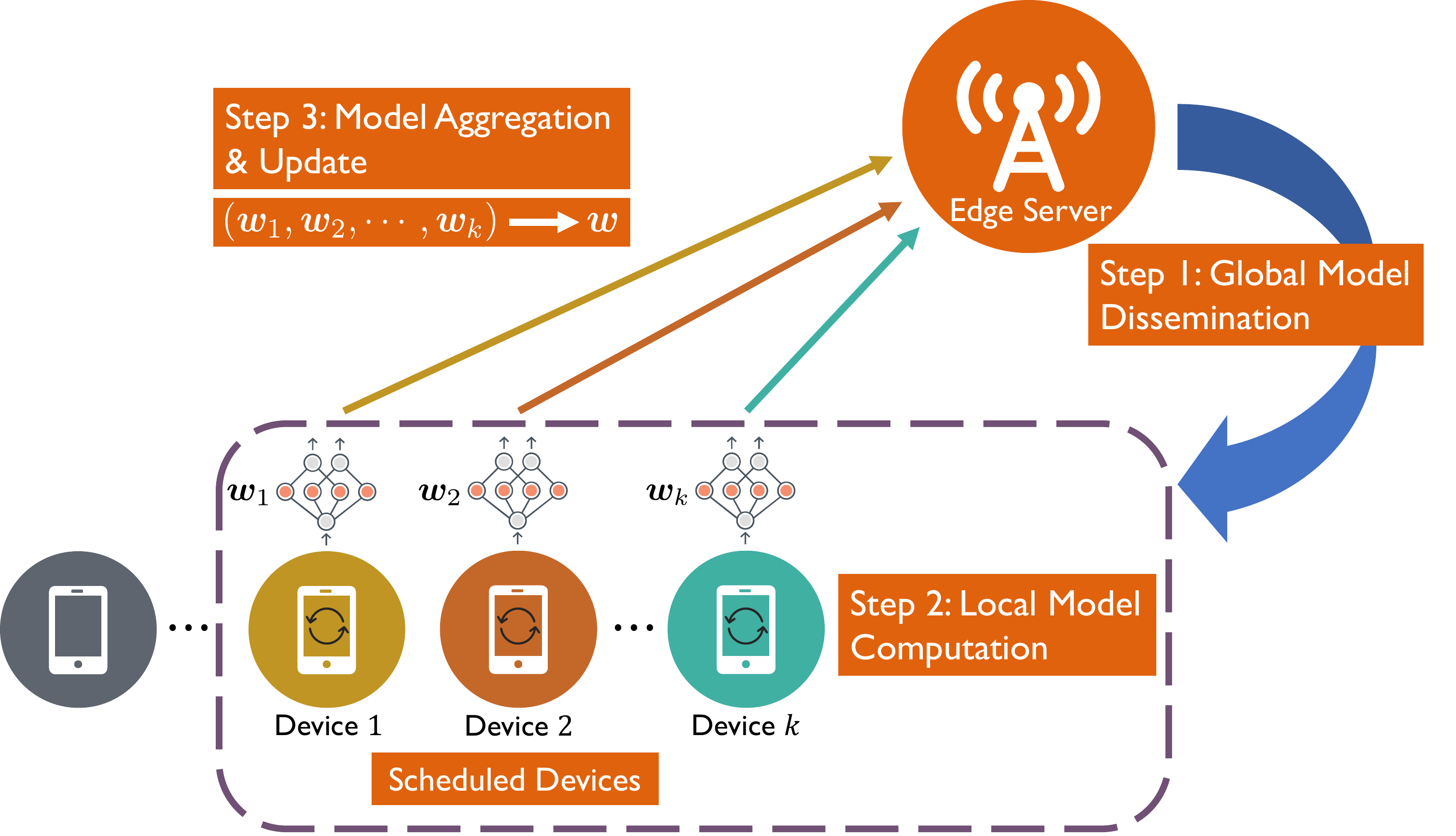}
	\vspace{-0.6cm}
	\caption{The canonical FL system consisting of one edge server and devices. }
	\label{FL_Byzantine}
\vspace{-0.5cm}
\end{figure}

\item [2)] \emph{Local Model Computation: }After each scheduled device $ k \in \mathcal{S}_t $ receives the current global model parameters, it performs local computation according to its own dataset to update the local model parameters $ \bm{w}_{k,t} $ \cite{mcmahan2017communication}. For simplification, we consider that each scheduled device updates its local model parameters via one-step gradient descent method \cite{shi2020joint, wang2020federated}, i.e.,
\begin{equation}
	\bm{w}_{k,t} = \bm{w}_{t-1} - \eta_{t} \frac{\sum_{i = 1}^{n_k}\nabla f(\bm{w}_{t-1};\bm{\xi}_{k,i})}{n_k}, 
	\label{update_rule}
\end{equation}
where $ \bm{w}_{t-1} $ is the received global model parameters, $ \eta_{t}$ is the learning rate, and $ \bm{w}_{k,t} $ is the generated local model parameters of device $k$.

\item [3)] \emph{Model Aggregation \& Update: }All scheduled devices upload their local models. Then the edge server uses a weighted sum method to aggregate these local models so as to generate new global model parameters, i.e.,
\begin{equation}
	\bm{w}_{t}= \frac{\sum_{k \in \mathcal{S}_t} \alpha_k \bm{w}_{k,t}}{\sum_{k \in \mathcal{S}_t} \alpha_k} ,
	\label{agg_rule}
\end{equation}
where $ \bm{w}_{t} $ is the updated global model parameters.  
\end{enumerate}

These three stages are repeated until the whole FL system attains a sufficiently trained model, as shown in Fig. \ref{FL_Byzantine}.

\subsection{Transmission Model}
We study the information exchange process between the edge server and devices over shared wireless MIMO networks, where the edge server equipped with $ N $ antennas orchestrates all single-antenna devices. The model parameters are encoded into digital signals to achieve reliable error-free transmission \cite{amiri2020federated}, such as polar code. Additionally, we consider the block flat fading channel and assume that one FL round could be completed within a communication block. The assumption is practically well justified when on-device FL models are typically light-weight under a few tens of thousands of parameters, whose time consumption is in the same order of channel coherent block \cite{liu2020privacy}. Therefore, it is reasonable to finish one round of FL training process within one communication block.

We first consider the uplink process, that all scheduled devices are concurrently transmitting their data streams. The received signal at the edge server can be expressed as
\begin{equation}
	\bm{y}_t = \sum_{k \in \mathcal{S}_t} \bm{h}_{k,t} \sqrt{p_{k,t}} x_{k,t} + \bm{n}_t,
\end{equation}
where $x_{k, t}$ is the transmitted signal symbol of device $k$, $p_{k, t}$ is the transmit power of device $k$, $\bm{h}_{k, t} \in \mathbb{C}^{N}$ is the wireless channel coefficient between the edge server and device $k$, and $\bm{n}_{t} \sim \mathcal{CN}(\bm{0}, \sigma^2 \bm{I})$ is the additive white Gaussian noise (AWGN) at the edge server, where $\sigma^2$ is noise power. Herein, $x_{k, t}$ is normalized with $\mathbb{E}\{|x_{k, t}|^2\} = 1$, and the noise power satisfies $\sigma^2 = B N_{0}$, where $B$ is the uplink bandwidth and $N_0$ is the noise power spectral density.

Due to the existence of interference and noise, the uplink process is the primary bottleneck of one computation and communication round. To improve the communication efficiency of the FL system, the linear receive beamforming technique is deployed at the edge server to decode $k$-th device's data stream, which is denoted as $\bm{m}_{k, t} \in \mathbb{C}^{N}$. Without loss of generality, we normalize the receive beamforming vector as $ \| \bm{m}_{k,t} \|_2^2 = 1$. Then, we have the signal-to-interference-plus-noise ratio (SINR) for the $k$-th device's data stream
\begin{equation}\label{SINR_express}
    {\sf SINR}^{\sf ul}_{k, t}(\bm{m}_{k, t}, \bm{p}_{t})
    = \frac{ p_{k, t} \left|\bm{m}_{k, t}^{\sf H} \bm{h}_{k, t}\right|^2 }{ \sum_{i \in \mathcal{S}_{t}/\{k\}} p_{i, t} \left|\bm{m}_{k, t}^{\sf H} \bm{h}_{i, t}\right|^2 + \sigma^2 },
\end{equation}
where $\bm{p}_{t} = [p_{1, t}, p_{2, t}, \ldots, p_{K, t}]^{\sf T}$ is the collection of transmit power of all devices. Note that ${\sf SINR}^{\sf ul}_{k, t}$ represents the possibly minimal SINR of device $k$ during the whole uplink transmission interval at round $t$.

\subsection{Latency Model}
Based on the three stages of the FL process in Section \ref{FL_Model}, the computation and communication latency can be mainly classified into three categories \cite{shi2020joint}, i.e.,
\begin{enumerate}
\item \textit{Downlink Broadcast Latency:}
In view of the fact that the edge server has relatively less stringent power constraint than devices and could occupy the whole downlink bandwidth to broadcast the global model, the downlink broadcast latency is negligible.

\item \textit{Local Computation Latency:} Since each scheduled device executes one-step update via gradient descent method based on its local dataset, the local computation latency of device $k$ is given by
\begin{equation}
    T_{k,t}^{\sf loc} = \frac{n_{k} R}{f_{k}^{\sf cap}},
\end{equation}
where $n_{k}$ is the size of local dataset at device $k$, $R$ is the number of processing unit (e.g., CPU or GPU) cycles for calculating one data sample, and ${f_{k}^{\sf cap}}$ is the computational capacity of device $k$, which is quantified by the frequency of the processing unit.

\item \textit{Uplink Transmission Latency:} Combining the possible minimal SINR expression in \eqref{SINR_express}, the uplink transmission rate of device $k$ can be expressed as
\begin{equation} \label{uplink_rate}
    r_{k, t}^{\sf ul} = B \operatorname{log}_{2}\left(1 + {\sf SINR}^{\sf ul}_{k, t} \right).
\end{equation}
Using $I$-bit number to represent model parameter, the uplink transmission latency of device $k$ is given by
\begin{equation}
    T_{k,t}^{\sf ul} = \frac{Id}{r_{k, t}^{\sf ul}} = \frac{Id}{B \operatorname{log}_{2}\left(1 + {\sf SINR}^{\sf ul}_{k, t} \right)},
\end{equation}
where $d$ is the dimension of model parameters. Herein, instead of adaptive rate transmission strategy, we employ the fixed rate transmission strategy for simplification, and choose the channel capacity of the worst case $\operatorname{log}_{2}\left(1 + {\sf SINR}^{\sf ul}_{k, t} \right)$ as the fixed transmission rate.
\end{enumerate}
On account of the synchronization requirement of the FL system and the limited length of wireless channel coherent block in practice, we expect to constrain the total latency of the $t$-th round, which is determined by the slowest device \cite{shi2020joint}
\begin{equation} \label{total_latency}
	T_t^{\sf sys}(\mathcal{S}_t,\{\bm{m}_{k,t}\} ,{\bm{p}_t}) = \underset{k \in \mathcal{S}_t}{\text{max}}(T_{k,t}^{\sf loc} + T_{k,t}^{\sf ul}).
\end{equation}
This indicates that each device could start its local computation as long as it receives the global model parameters, and then uploads its local model parameters as long as it accomplishes its local computation.

\section{Convergence Analysis and Problem Formulation}
\subsection{Convergence Analysis}
\label{Convergence_Analysis}
We establish the convergence analysis of the FL process based on the following assumptions, which have been made in the works \cite{wang2020federated, li2019convergence}.
\begin{assumption} \label{L_smooth}
	\emph{($L$-smoothness):} The differentiable function $F(\bm{w})$ is smooth with a positive constant $L$, i.e., for all $\bm{v}$ and $\bm{w}$, we have
	\begin{equation}
		F(\bm{v}) \leq F(\bm{w}) + (\bm{v}-\bm{w})^{\sf T} \nabla F(\bm{w}) + \frac{L}{2}\|\bm{v}-\bm{w}\|_2^2.
	\end{equation} 
\end{assumption}

\begin{assumption} \label{gradient_bound}
\emph{(Bounded local gradients):} The local gradients at all devices are uniformly bounded, i.e., there exist constants $\kappa \geq 0$ such that for all $\bm{w}$ and $\bm{\xi}$,
\begin{equation}
	\|\nabla f(\bm{w}; \bm{\xi}) \|_2^2 \leq \kappa.
\end{equation}
\end{assumption}

\begin{theorem} \label{thm1}
Suppose that \textbf{Assumption \ref{L_smooth}} and \textbf{\ref{gradient_bound}} hold, then given the collection of scheduling results $\{\mathcal{S}_t\}$ and setting the learning rate to be $0 < \eta_{t} \equiv \varsigma \leq \frac{1}{L}$, the average norm of global gradients after $\tau$ rounds is upper bounded by
\begin{equation} \label{theorem_1}
\begin{aligned}
\frac{1}{\tau} \sum_{t=0}^{\tau-1} \| \nabla F(\bm{w}_{t-1} ) \|_2^2 \leq & \frac{2 \left(F(\bm{w}_0) - F(\bm{w}^*) \right) }{\varsigma \tau} \\
& + \underbrace{\frac{4\kappa}{\tau} \sum_{t=0}^{\tau-1} \left(1 - \sum_{k \in \mathcal{S}_t} \alpha_k \right)^2}_{ g(\{\mathcal{S}_t\}) },
\end{aligned}
\end{equation}
where $\bm{w}^*$ is the globally optimal solution for \eqref{main_problem}.
\end{theorem}
\begin{proof}
Please refer to Appendix \ref{proof_Theorem_1}.
\end{proof}

\subsection{Problem Formulation}
Based on \textbf{Theorem \ref{thm1}}, the convergence optimality gap is dominated by the second term $g(\{\mathcal{S}_t\})$. We now formulate the system optimization problem to minimize the convergence gap $g(\{\mathcal{S}_t\})$ under system latency constraint $T^{\sf thr}$ and sum power constraint $P_{\sf sum}$, which is written as
\begin{equation} \label{original_problem}
\begin{aligned}
\underset{\{\mathcal{S}_{t}\}, \{\bm{m}_{k, t}\}, \{\bm{p}_t\} }{\text{minimize}} \  & g\left(\left\{\mathcal{S}_{t}\right\}\right)\\
\text{subject to} \quad \  & T_{t}^{\sf sys} (\mathcal{S}_{t}, \{\bm{m}_{k, t}\}, \bm{p}_{t}) \leq T^{\sf thr},  \forall t,\\
&\|\bm{m}_{k, t}\|_2^2 = 1, p_{k, t} \geq 0, \forall k \in \mathcal{S}_{t}, \forall t,\\
&\ \sum_{k \in \mathcal{S}_{t}} p_{k, t} \leq P_{\sf sum},  \forall t,
\end{aligned}
\end{equation}
\begin{algorithm}
\label{scheduling problem}
\caption{Device scheduling of one FL round}
\KwIn{Channel coefficient $\bm{h}_k$, noise power $\sigma^2$, SINR requirement $\gamma^{\sf thr}_{k}$, device power constraint $P_{\sf sum}$.}
Solve problem \eqref{final_problem}, and sort $\bm{s}$ in ascending order. \\
Initialize the device scheduling set $\mathcal{S}  = \mathcal{S}^{\sf tmp} = \varnothing$. \\
\While{ $\mathcal{S} \neq \mathcal{K}$ }
{
    Add a new device to $\mathcal{S}^{\sf tmp}$ with the lowest $s_{k}$.\\
    Test if the SINR and sum power constraints are feasible for $\mathcal{S}^{\sf tmp}$ via \textbf{Algorithm \ref{feasible_check}}.\\
    If False, terminate the loop.\\
    $\mathcal{S} = \mathcal{S}^{\sf tmp}$
}
\textbf{return} $\mathcal{S}$ and corresponding $\bm{p}$ and $\{\bm{m}_{k}\}$
\end{algorithm}
where $\{\mathcal{S}_t\}$, $\{\bm{m}_{k,t}\}$ and $\{\bm{p}_t\}$ represent the collection of scheduling results, receive beamforming vectors and devices' transmit power during the total $\tau$ rounds, respectively. Note that the edge server must wait for the local models of all scheduled devices before updating the global model, so the system latency constraint $T^{\sf thr}$ plays a key role in the FL performance \cite{chen2020joint}. Since constraints of problem \eqref{original_problem} are independent of round $t$ and $g(\{\mathcal{S}_t\})$ is a decreasing function with respect to $\sum_{k\in\mathcal{S}_t} \alpha_k$, we could decouple problem \eqref{original_problem} into $\tau$ one-round sub-problem, and we solve the following one-round system optimization problem
\begin{equation} \label{transfer2}
\begin{aligned}
\underset{ \mathcal{S}, \{\bm{m}_{k}\}, \bm{p} }{\text{maximize}} \  & \sum_{k \in \mathcal{S}} \alpha_{k}\\
\text{subject to} \  & T^{\sf sys} (\mathcal{S}, \{\bm{m}_{k}\}, \bm{p}) \leq T^{\sf thr},\\
&\|\bm{m}_{k}\|_2^2 = 1, p_{k} \geq 0, \forall k \in \mathcal{S},\\
&\ \sum_{k \in \mathcal{S}} p_{k} \leq P_{\sf sum},
\end{aligned}
\end{equation}
where the subscript $t$ is omitted for brevity. By substituting \eqref{total_latency} into problem \eqref{transfer2}, we obtain
\begin{equation} \label{transfer3}
\begin{aligned}
\underset{ \mathcal{S}, \{\bm{m}_{k}\}, \bm{p} }{\text{maximize}} \  & \sum_{k \in \mathcal{S}} \alpha_{k}\\
\text{subject to} \  & {\sf SINR}^{\sf ul}_{k} \geq {\gamma}_{k}^{\sf thr} , \forall k \in \mathcal{S},\\
&\|\bm{m}_{k} \|_2^2 = 1, p_{k} \geq 0, \forall k \in \mathcal{S},\\
&\ \sum_{k \in \mathcal{S}} p_{k} \leq P_{\sf sum},
\end{aligned}
\end{equation}
where
$$
{\gamma}_{k}^{\sf thr} = 2^{ r_{k} } - 1 \text{ and } r_{k} = \frac{ I d }{ B ( T^{\sf thr} - T^{\sf loc}_{k}) }.
$$
Problem \eqref{transfer3} is challenging to solve due to the combinatorial optimization variable $\mathcal{S}$, the sparse objective function and the non-convex SINR constraints. To tackle this issue, we shall exploit the uplink–downlink duality of MIMO systems in the next section.

\begin{algorithm}
 \label{feasible_check}
 \caption{Feasibility test in the uplink transmission}
 \textbf{Initialize:}Arbitrary $\bm{p}^{(0)}$ such that $\sum_{k \in \mathcal{S}} p_k^{(0)} = P_{\sf sum}$.\\
 \Repeat{$\bm{p}$ {\rm convergence}}{
 		In the $l$-the iteration, update $\bm{p}^{(l-1)}$ according to 
 		\begin{equation}
 			\tilde{p}_k = \frac{\gamma_k^{\sf thr}}{\bm{h}_k^{\sf H}\bm{\Sigma}_k^{-1} \bm{h}_k}, \forall k \in \mathcal{S},
 		\end{equation}
 		where
 		\begin{equation}
 			\bm{\Sigma}_k = \sum_{ i \in \mathcal{S}, i \neq k} p_i^{(l-1)} \bm{h}_i \bm{h}_i^{\sf H} + \sigma^2\bm{I}.
 		\end{equation}
 		Normalize $\tilde{p}_k$ according to $p_k^{(l)} = \frac{P_{\sf sum}}{\sum_{k \in \mathcal{S}} \tilde{p}_k}\tilde{p}_k$.
 		}
\Return{ {\rm Boolean value of }$\sum_{k \in \mathcal{S}} p_k \leq P_{\sf sum}$ }
\end{algorithm}

\section{System Optimization} \label{System Optimization}
In this section, we utilize the equivalence between the uplink and downlink device scheduling problems \cite{zhao2015user}, where the uplink SINR constraints for all devices and the sum power constraint is converted into a dual downlink constraints, as presented in the following lemma.
\begin{lemma}
A scheduling set $\mathcal{S}$ can satisfy SINR requirements $\{ \gamma_{k}^{\sf thr} \}$ and sum power constraint $P_{\sf sum}$ of \eqref{transfer3} in the uplink transmission if and only if there exist dual transmit beamforming vectors $\hat{\bm{m}}_{k} \in \mathbb{C}^{N}$ such that
\begin{equation} \label{downlink_constraint}
\left\{ \begin{aligned}
&{\sf SINR}^{\sf dl}_{k} \geq {\gamma}_{k}^{\sf thr}, \forall k \in \mathcal{S}\\
&\sum_{k\in \mathcal{S}} \|\hat{\bm{m}}_{k}\|_2^{2} \leq \frac{P_{\sf sum}}{\sigma^2}
\end{aligned}\right.,
\end{equation}
where dual downlink SINR for the $k$-th device is defined as
\begin{equation}
    {\sf SINR}^{\sf dl}_{k} = \frac{ \left|\hat{\bm{m}}_{k}^{\sf H} \bm{h}_{k}\right|^2 }{ \sum_{i \in \mathcal{S}/\{k\}} \left|\hat{\bm{m}}_{i}^{\sf H} \bm{h}_{k}\right|^2 + 1}.
\end{equation}
\end{lemma}
Therefore, the uplink device scheduling problem \eqref{transfer3} can be equivalently reformulated as the dual downlink problem
\begin{equation}\label{transfer4}
\begin{aligned}
\underset{ \mathcal{S}, \{\hat{\bm{m}}_{k}\} }{\text{maximize}} \  & \sum_{k \in \mathcal{S}} \alpha_{k}\\
\text{subject to} \  & \frac{\text{Re}( \hat{\bm{m}}_{k}^{\sf H} \bm{h}_{k} )}{\sqrt{{\gamma}_{k}^{\sf thr}}}  \geq \sqrt{\sum_{i \in \mathcal{S}/\{k\}} \left|\hat{\bm{m}}_{i}^{\sf H} \bm{h}_{k}\right|^2 + 1},\\
&\text{Im}(\hat{\bm{m}}_{k}^{\sf H} \bm{h}_{k})  =  0, \forall  k \in \mathcal{S},\\
&\sum_{k\in \mathcal{S}} \|\hat{\bm{m}}_{k}\|_2^{2} \leq \frac{P_{\sf sum}}{\sigma^2},
\end{aligned}
\end{equation}
where the downlink SINR constraints are equivalently rewritten as the second order cone constraints, since the phase of $\hat{\bm{m}}_{k}$ will not change the objective function and constraints \cite{shi2014group}. However, problem \eqref{transfer4} is still difficult to solve because of the combinatorial optimization variable $\mathcal{S}$ and the non-convex and non-smooth objective function. By introducing the auxiliary variable $\bm{s}$ and applying reweighted $\ell_{1}$ minimization technique, we relax problem \eqref{transfer4} to
\begin{figure*}[!ht]
\centering{
\vspace{-0.1cm}
    \subfigure[Training loss vs. Round]{\includegraphics[width=0.32\textwidth]{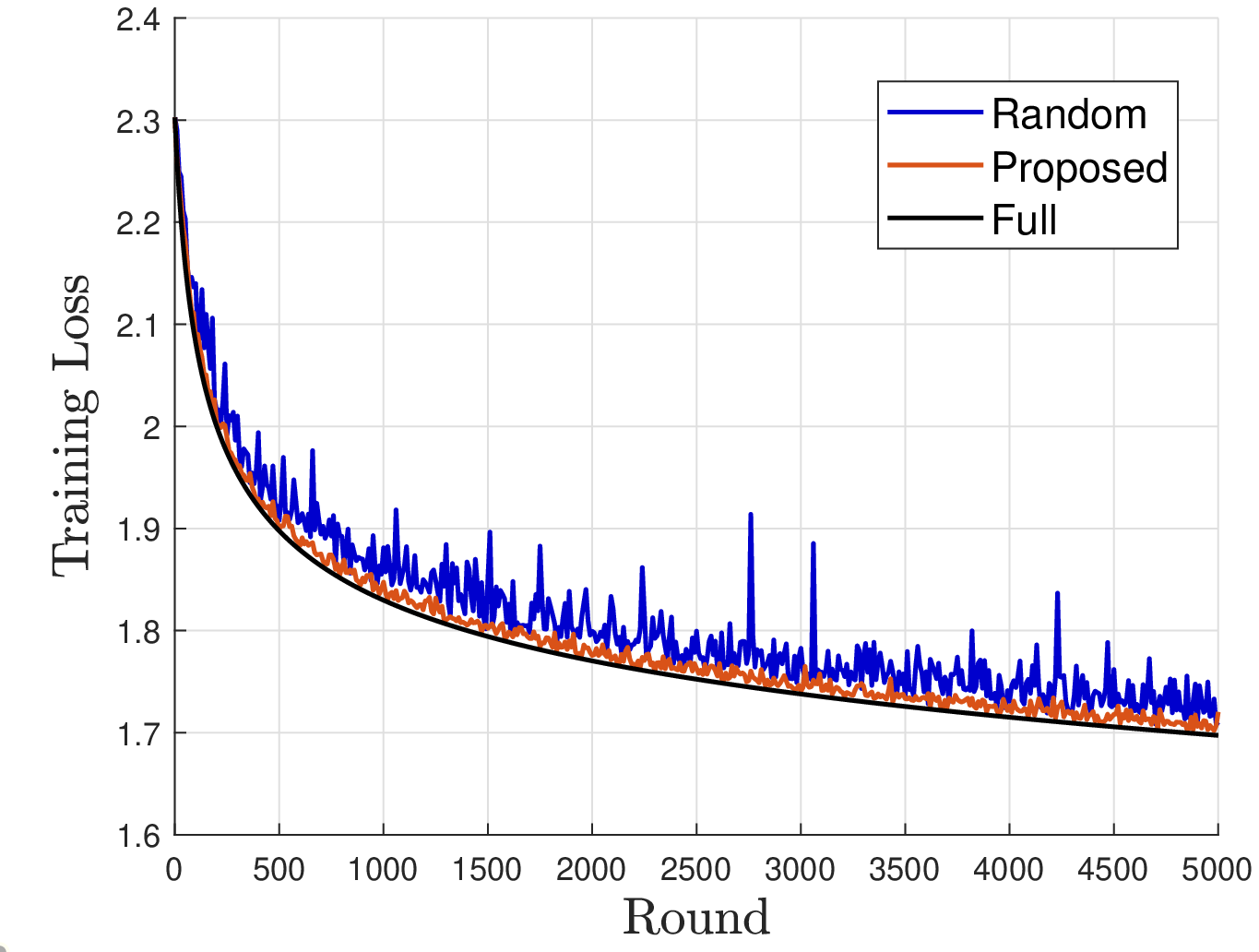}}
    \subfigure[Training accuracy vs. Round]{\includegraphics[width=0.32\textwidth]{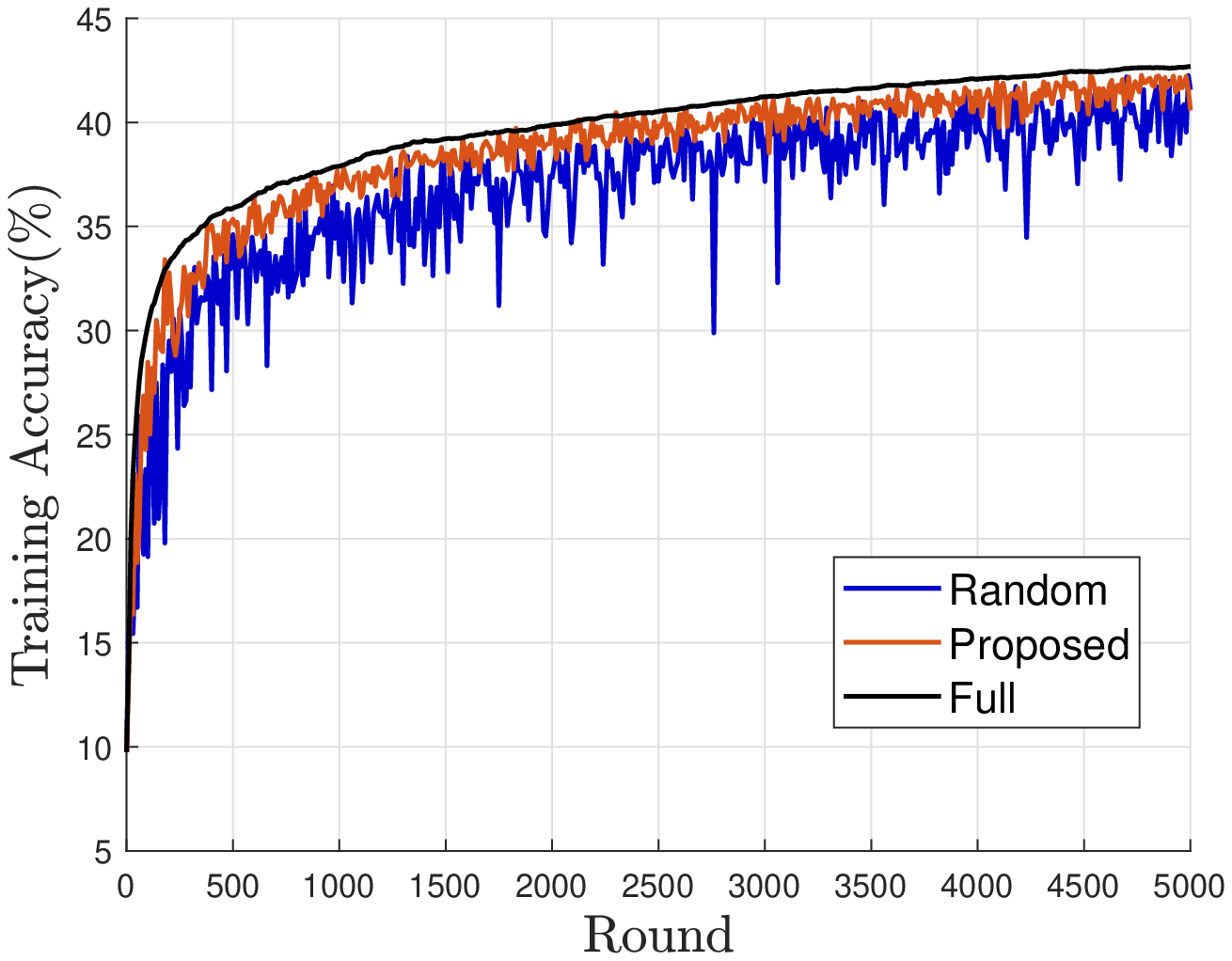}}
    \subfigure[Weighted sum of scheduled devices vs. Round]{\includegraphics[width=0.32\textwidth]{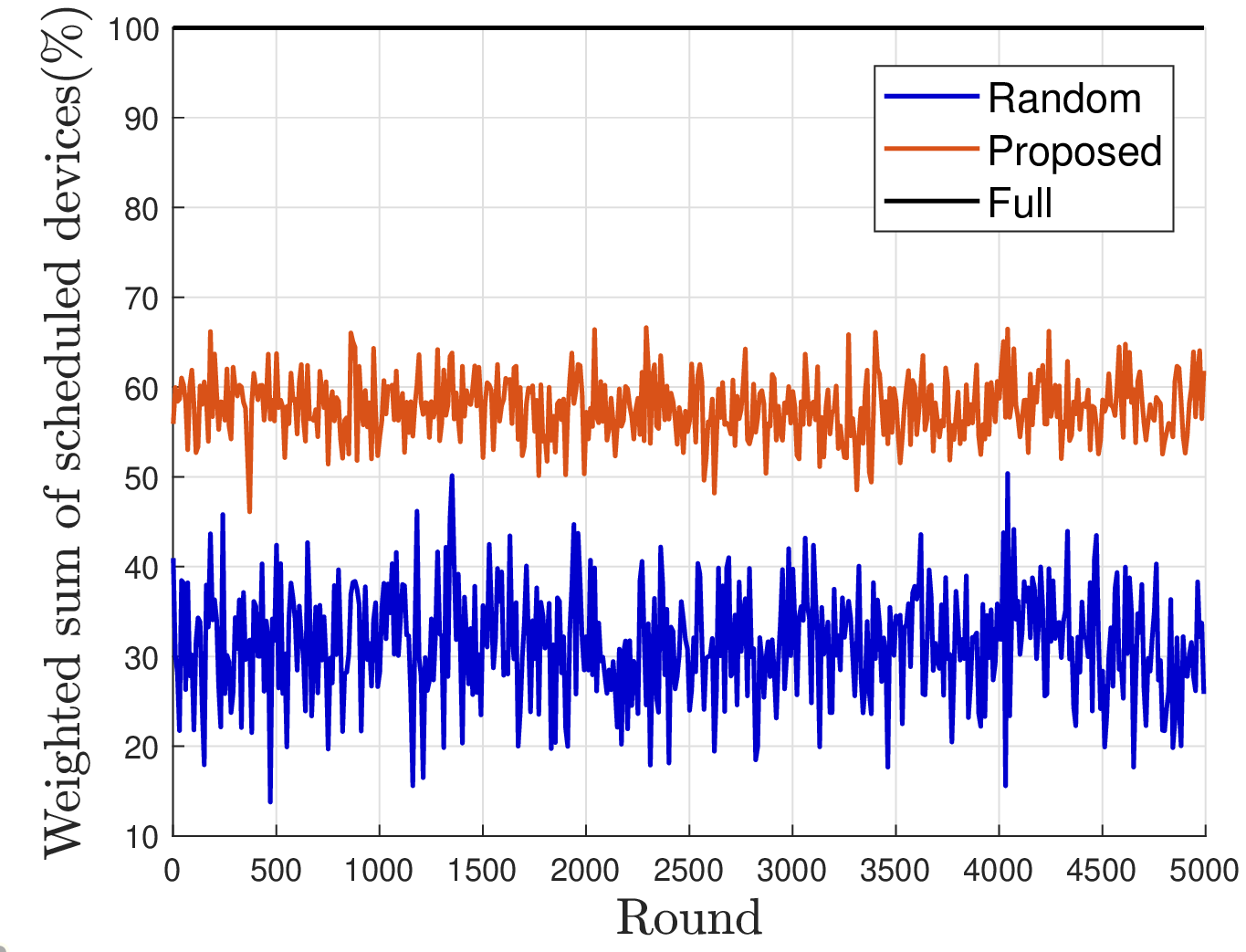}}
    \vspace{-0.1cm}
    \caption{The convergence results of the FL system over shared wireless MIMO networks under different device scheduling policies.}
    \label{Simulation_result2}
}
\vspace{-0.4cm}
\end{figure*}
\begin{equation}\label{final_problem}
\begin{aligned}
\underset{ \bm{s}, \{\hat{\bm{m}}_{k}\}}{\text{minimize}} \  & \sum_{k \in \mathcal{K}} \alpha_{k} s_{k}\\
\text{subject to} \  & \frac{\text{Re}( \hat{\bm{m}}_{k}^{\sf H} \bm{h}_{k} )}{\sqrt{{\gamma}_{k}^{\sf thr}}} \! + \! s_{k}  \geq \sqrt{\sum_{i \in \mathcal{K}/\{k\}} \left|\hat{\bm{m}}_{i}^{\sf H} \bm{h}_{k}\right|^2 \! + \! 1},\\
&\text{Im}(\hat{\bm{m}}_{k}^{\sf H} \bm{h}_{k}) = 0, \forall  k \in \mathcal{K},\\
&\sum_{k\in \mathcal{K}} \|\hat{\bm{m}}_{k}\|_2^{2}  \leq  \frac{P_{\sf sum}}{\sigma^2}, \bm{s} \geq \bm{0}.
\end{aligned}
\end{equation}
The proposed device scheduling policy is summarized in \textbf{Algorithm \ref{scheduling problem}}. More specifically, we first solve problem \eqref{final_problem} to get the device priority $\bm{s}$, and then iteratively add one device to the scheduling set $\mathcal{S}$ based on the device priority $\bm{s}$ if feasibility test in \textbf{Algorithm \ref{feasible_check}} could be passed \cite{cai2011unified}. Finally, we obtain the optimization solutions about the scheduling set $\mathcal{S}$ and the power allocation $\bm{p}$. The optimal receive beamforming vectors $\{\bm{m}_{k}\}$ that maximize the uplink SINR are the corresponding minimum-mean-square-error (MMSE) filters \cite{yu2007transmitter}, which are obtained in closed form as
\begin{equation}
    \bm{m}_{k} = \frac{ \left( \sigma^{2}\bm{I} + \sum_{i \in \mathcal{S}} p_{i} \bm{h}_{i} \bm{h}^{\sf H}_{i} \right)^{-1} \bm{h}_{k} }{\left\| \left( \sigma^{2}\bm{I} + \sum_{i \in \mathcal{S}} p_{i} \bm{h}_{i} \bm{h}^{\sf H}_{i} \right)^{-1} \bm{h}_{k} \right\|_2^{2}}, \forall k \in \mathcal{S}.
\end{equation}

\section{Simulation Results}
In this section, we evaluate the convergence gap of the FL system under the different device scheduling policies by numerical experiments. We consider the image classification task on the CIFAR-10 test dataset, which consists of $10000$ $32 \times 32$ RGB color images with $10$ classes. We adopt the \textit{Multinomial Logistic Regression} model ($d = 32\times32\times3\times10$ parameters) to classify the target dataset with the learning rate $\eta_{t} = 5\times 10^{-3}$ for all fixed $\tau = 5000$ rounds. Each parameter is stored with $I = 32$ bits to guarantee the numerical precision, and each data sample can be handled within $\frac{R}{f^{cap}_{k}} = 10^{-4}${\rm s} for all devices. We consider a FL system consisting of one edge server equipped with $N = 4$ antennas and $50$ single-antenna devices, and the devices are uniformly located in a region enclosed between the inner radius of $50$ meters and outer radius of $250$ meters. We assign the non-iid dataset to each device. In particular, we sort the image data samples according to their class, and divide them into $50$ disjoint sub-datasets with different sample sizes. The channel coefficient $\bm{h}_{k}$ at a distance of $d_{k}$ meters is generated as $\bm{h}_{k} = \sqrt{\beta_{k}} \tilde{\bm{h}}_{k}$, where path loss $\beta_{k} = -35.3 - 37.6 \log_{10}(d_{k})$ (in {\rm dB}) and $\tilde{\bm{h}}_{k}$ is independently generated via Rayleigh fading $\mathcal{CN}(\bm{0}, \bm{I})$. The noise power spectral density $N_0$, the uplink bandwidth $B$, the system latency constraint $T^{\sf thr}$ and the sum power constraint $P_{\rm sum}$ are set to $-174${\rm dBm/Hz}, $10${\rm MHz}, $1${\rm s} and $30{\rm mW}$, respectively \cite{shi2020joint}.

Fig. \ref{Simulation_result2} shows the convergence results of the FL system over shared wireless MIMO networks under different device scheduling policies. To better illustrate the simulation results, we make plot using one point for every $10$ result data samples (i.e., $1, 11, 21, \ldots$). To be specific, `Full' scheduling policy represents that all devices participate in the training process at each FL round without the wireless resource constraints. This is equivalent to conventional centralized ML scheme. `Proposed' scheduling policy represents that we select a subset of devices using the proposed \textbf{Algorithm \ref{scheduling problem}} at each FL round. `Random' scheduling policy represents that we randomly select a subset of devices at each FL round. Such randomly selected subset is required to pass feasibility test in \textbf{Algorithm \ref{feasible_check}} and thereby satisfies the wireless resource constraints. The simulation results manifest that the proposed scheduling policy could allow more weighted devices to participate in each training round, and consequently brings smaller convergence optimality gap and higher model training accuracy than the random scheduling policy. Observation from numerical experiments further verify the convergence analysis result in Section \ref{Convergence_Analysis}, i.e., the convergence optimality gap decreases as the weighted sum of scheduled devices increases at each FL round.

\section{Conclusion}
In this paper, we proposed a novel joint device scheduling and receive beamforming design approach for a delay-aware MIMO FL system. Specifically, we established the convergence analysis of the FL system, and then joint optimized the receive beamforming design and the device scheduling to maximally schedule weighed devices so as to further reduce the convergence optimality gap. Numerical results demonstrated that the proposed device scheduling policy could substantially enhance the learning performance of the FL process compared with the random scheduling policy.

\begin{appendices}
\section{Proof of Theorem \ref{thm1}} \label{proof_Theorem_1}
Combining the update rule \eqref{update_rule} of device $k$ and aggregation rule \eqref{agg_rule} of the edge server, we have
\begin{equation}
\begin{aligned}
    \bm{w}_{t} =& \frac{ \sum_{k \in \mathcal{S}_{t}} \alpha_{k} \left(\bm{w}_{t-1} - \eta_{t}  \frac{\sum_{i= 1}^{ n_{k} }\nabla f(\bm{w}_{t-1}; \bm{\xi}_{k,i} ) }{n_{k}} \right) } {\sum_{k \in \mathcal{S}_{t}} \alpha_{k} } \\
    =& \bm{w}_{t-1} - \eta_{t} \frac{ \sum_{k \in \mathcal{S}_{t}} \frac{\alpha_{k}}{ n_{k} } \sum_{i= 1}^{ n_{k} }\nabla f(\bm{w}_{t-1}; \bm{\xi}_{k,i} ) } {\sum_{k \in \mathcal{S}_{t}} \alpha_{k} } \\
    =& \bm{w}_{t-1} - \eta_{t} \left( \nabla F(\bm{w}_{t-1}) + \bm{e}_{t} \right),
\end{aligned}
\label{wt_iter}
\end{equation}
where the residual term $\bm{e}_{t}$ is defined as
\begin{equation}
\begin{aligned}
    \bm{e}_{t}
    =& \frac{ 1 } {\sum_{k \in \mathcal{S}_{t}} \alpha_{k} }  \sum_{k \in \mathcal{S}_{t}} \frac{\alpha_{k}}{ n_{k} } \sum_{i= 1}^{ n_{k} }\nabla f(\bm{w}_{t-1}; \bm{\xi}_{k,i} ) \\
    &- \underbrace{\sum_{k \in \mathcal{K}} \frac{\alpha_{k}}{ n_{k} } \sum_{i= 1}^{ n_{k} }\nabla f(\bm{w}_{t-1}; \bm{\xi}_{k,i} ) }_{\nabla F(\bm{w}_{t-1})}\\
    =& \frac{ \sum_{k \notin \mathcal{S}_{t}} \alpha_{k} } {\sum_{k \in \mathcal{S}_{t}} \alpha_{k} }  \sum_{k \in \mathcal{S}_{t}} \frac{\alpha_{k}}{ n_{k} } \sum_{i= 1}^{ n_{k} }\nabla f(\bm{w}_{t-1}; \bm{\xi}_{k,i} )\\
    & - \sum_{k \notin \mathcal{S}_{t}} \frac{\alpha_{k}}{ n_{k} } \sum_{i= 1}^{ n_{k} }\nabla f(\bm{w}_{t-1}; \bm{\xi}_{k,i} ).\\
\end{aligned}
\end{equation}
Under \textbf{Assumption \ref{gradient_bound}}, we apply the norm inequality to $\bm{e}_{t}$
\begin{equation}\label{residual}
\begin{aligned}
    \left\|  \bm{e}_{t}  \right\|_2^{2} \leq& \left(  \frac{ \sum_{k \notin \mathcal{S}_{t}} \alpha_{k} } {\sum_{k \in \mathcal{S}_{t}} \alpha_{k} } \sum_{k \in \mathcal{S}_{t}} \frac{\alpha_{k}}{ n_{k} } \sum_{i= 1}^{ n_{k} } \left\| \nabla f(\bm{w}_{t-1}; \bm{\xi}_{k,i} ) \right\|_2 \right. \\
    & \left. + \sum_{k \notin \mathcal{S}_{t}} \frac{\alpha_{k}}{ n_{k} } \sum_{i= 1}^{ n_{k} } \left\|\nabla f(\bm{w}_{t-1}; \bm{\xi}_{k,i} ) \right\|_2 \right )^{2}\\
    \leq&  4 \kappa \left(  \sum_{k \notin \mathcal{S}_{t}} \alpha_{k}  \right)^{2}  .
\end{aligned}
\end{equation}
Under \textbf{Assumption \ref{L_smooth}} and $0 < \eta_{t} \equiv \varsigma \leq \frac{1}{L}$, we have
\begin{equation} \label{tmp_iter}
\begin{aligned}
    F(\bm{w}_{t}) \leq& F(\bm{w}_{t-1}) + (\frac{\varsigma^2 L}{2} - \varsigma ) \left\| \nabla F(\bm{w}_{t-1}) \right \|_2^{2} + \frac{\varsigma^2 L}{2} \left\| \bm{e}_{t} \right\|_2^{2} \\
    &+ (\varsigma - \varsigma^2 L) \nabla F(\bm{w}_{t-1})^{\sf T}  \bm{e}_{t} \\
    \leq& F(\bm{w}_{t-1}) -\frac{\varsigma}{2} \left\| \nabla F(\bm{w}_{t-1}) \right \|_2^{2} + \frac{\varsigma}{2} \left\| \bm{e}_{t} \right\|_2^{2}.
\end{aligned}
\end{equation}
Substitute \eqref{residual} into \eqref{tmp_iter}, we obtain
\begin{equation}\label{one_times}
\begin{aligned}
    \left\|  \nabla F(\bm{w}_{t-1}) \right\|_2^{2} \leq& \frac{2 \left( F(\bm{w}_{t-1}) - F(\bm{w}_{t}) \right) }{ \varsigma } + 4 \kappa \left( \sum_{k \notin \mathcal{S}_{t}} \alpha_{k}  \right)^{2} .
\end{aligned}
\end{equation}
Then, summing both sides of \eqref{one_times} for $t \in \{1, \ldots, \tau\}$ and combinging $F(\bm{w}_{\tau}) \geq F(\bm{w}^{*})$, we can obtain \textbf{Theorem \ref{thm1}}.
\end{appendices}

\bibliography{Reference} 
\bibliographystyle{ieeetr}

\end{document}